\documentclass[twocolumn,secnumarabic,amssymb, nobibnotes, aps, prd, superscriptaddress]{revtex4-2}
\usepackage{amsmath,amssymb,amsfonts}
\usepackage{amsthm}
\usepackage{bm}
\usepackage{braket}

\usepackage{graphicx}
\usepackage{epstopdf}
\usepackage{booktabs}

\usepackage{float}
\usepackage[justification=centering]{caption}
\usepackage{subcaption}

\usepackage[ruled,vlined]{algorithm2e}
\usepackage[title]{appendix}

\usepackage{adjustbox}
\usepackage{textcomp}
\usepackage{mathrsfs}
\usepackage{booktabs}

\usepackage[
bookmarks=true,
colorlinks=true,
linkcolor=blue,
urlcolor=blue,
citecolor=blue,
bookmarks=true,
hyperindex=true
]{hyperref}

\makeatletter

\newcommand{\Rmnum}[1]{\expandafter\@slowromancap\romannumeral #1@}
\makeatother

\begin{document}

\title{Local distinguishability of six bipartite orthogonal product states}

\author{Guang-Bao \surname{Xu}}
\affiliation{College of Computer Science and Engineering, Shandong University of Science and Technology, Qingdao, 266590, China}

\author{Hua-Kun \surname{Wang}}
\affiliation{College of Computer Science and Engineering, Shandong University of Science and Technology, Qingdao, 266590, China}

\author{Yu-Guang \surname{Yang}}
\affiliation{College of Cyberspace Science and Technology, Beijing University of Technology, Beijing, 100124, China}

\author{Dong-Huan \surname{Jiang}}
\email{donghuan\_jiang@163.com}
\affiliation{College of Mathematics and Systems Science, Shandong University of Science and Technology, Qingdao, 266590, China}

\date{\today}

\begin{abstract}
It is necessary to investigate the local distinguishability of orthogonal quantum state sets, as their adoption in protocol design helps diminish quantum state transmission and cut operational costs. In this paper, we explore the local distinguishability of six orthogonal product states (OPSs) on any bipartite quantum system. We classify different sets of six bipartite OPSs into eight categories by using the vectors of the numbers of pairwise orthogonality relations, where any two states are orthogonal on only one subsystem within each set. We find that these eight categories contain a total of 78 distinct cases, all but five of which are perfectly distinguishable via local operations and classical communication (LOCC). Furthermore, we discuss the local distinguishability of those five distinct cases in detail. Our work explicitly characterizes the local distinguishability of six bipartite OPSs.
	
\end{abstract}

\maketitle

\section{\label{sec1}Introduction\protect}
In a quantum protocol, the constituent particles of a composite quantum state may be distributed among spatially separated participants. By identifying the target quantum state through local operations and classical communication (LOCC), participants can not only reduce the overhead of quantum state transmission, but also substantially cut down the implementation cost of the quantum protocol. A set of orthogonal composite quantum states is locally distinguishable whenever its elements can be perfectly distinguished by LOCC; otherwise, it is locally indistinguishable or nonlocal. 

Typically, it is widely held that the local indistinguishability of quantum states is inherently linked to entanglement. In 1999, Bennett et al. ~\cite{1} constructed a set of nine bipartite orthogonal product states (OPSs) and demonstrated that the set is locally indistinguishable. This phenomenon, referred to as nonlocality without entanglement~\cite{1}, overturns the prevailing cognitive perception. Inspired and motivated by the work of Bennett et al., numerous relevant achievements \cite{J2000,J2002,DiVincenzo2003,Feng2009,Duan2009,Duan2014} have since emerged. Walgate et al. \cite{J2000} proved that any two orthogonal pure states can be perfectly identified by LOCC. Furthermore, walgate et al. \cite{J2002} gave a necessary and sufficient conditions for a set of $2\otimes2$ quantum states to be exactly distinguishable. DiVincenzo et al. \cite{DiVincenzo2003} proved that any set of no more than three OPSs on a multipartite system or no more than four OPSs on a bipartite system is locally distinguishable. Feng et al. gave structure characterizations for all feasible sets of OPSs that cannot be perfectly distinguished by LOCC on $3\otimes3$ and $2\otimes2\otimes2$ quantum systems. All these works have substantially enriched the theoretical research concerning the local distinguishability of quantum state sets.

Since Bennett et al. gave two methods to construct nonlocal sets of OPSs on $3\otimes3$ and $2\otimes2\otimes2$ quantum systems, the construction methods of nonlocal state sets in high-dimensional systems have drawn considerable research interest. Several constructions for nonlocal OPS sets in high-dimensional bipartite systems have been proposed in Refs. \cite{2,5,6,7,Cao2023}. Compared with bipartite cases, the constructions for nonlocal OPS sets in high-dimensional multipartite systems are more complex. Xu et al. \cite{8} gave a method to construct a nonlocal set of OPSs in high-dimensional multipartite systems via cyclic permutation of subsystems. Subsequently, numerous methods \cite{9,10,11,13} for constructing nonlocal sets of multipartite OPSs have been proposed one after another. In 2019, Halder et al. \cite{Halder2019} proposed the concept of strong quantum nonlocality without entanglement. Later, Many constructions of OPS sets with strong nonlocality \cite{Yuan2020,Shi2020,WangYL2021,FSHI2022,Fshi2022,Gao2023} have been proposed. These works demonstrate that it is of great significance to explore the construction methods of nonlocal sets of OPSs.

While much work has focused on the construction of nonlocal OPS sets, the exploration of local distinguishability of OPS sets with fixed cardinality is likewise of great academic significance. In 2025, Xu et al. \cite{GBX2026} gave the local distinguishability of four OPSs in arbitrary multipartite systems. Furthermore, they \cite{XHWYJ} resolved the local distinguishability problem of five OPSs in bipartite and tripartite quantum systems recently. A natural question now arises: what about the local distinguishability of six OPSs in an arbitrary bipartite system?

In this paper, we investigate the local distinguishability of  six bipartite OPSs, where any two of these OPSs are orthogonal on only one subsystem. By the vectors of the numbers of pairwise orthogonality relations, We first classify different sets of six bipartite OPSs into eight categories. Subsequently, we conduct a case analysis for every category. Among all the feasible 78 cases, only five cannot be perfectly distinguished by LOCC. Our results refine the theory of quantum nonlocality and offer insights for further research. 

The rest of this paper is organized as follows. Some foundational concepts and essential lemmas are given in Sec.~\ref{sec2}. In Sec.~\ref{sec3}, we systematically give three theorems, which will be used in the proof of local distinguishability of six bipartite OPSs. In Sec.~\ref{sec4}, We give the local distinguishability of six bipartite OPSs and provide a detailed proof. We discuss five special cases of six bipartite OPSs that cannot be locally distinguished in Sec.~\ref{sec5}. A brief conclusion is given in Sec.~\ref{sec6}.

\section{\label{sec2}Preliminaries}
\theoremstyle{remark}
\newtheorem{definition}{\indent Definition}
\newtheorem{lemma}{\indent Lemma}
\newtheorem{theorem}{\indent Theorem}
\newtheorem{corollary}{\indent Corollary}
\def\QEDclosed{\mbox{\rule[0pt]{1.3ex}{1.3ex}}}
\def\QED{\QEDclosed}
\def\proof{\indent{\em Proof}.}
\def\endproof{\hspace*{\fill}~\QED\par\endtrivlist\unskip}

This section presents the definitions and lemmas required for what follows.

\begin{definition}\label{def1}\cite{Bondy2008}
	(Undirected graph) An undirected graph $G$ is a pair $(V,\,E)$ where $V$ is a finite, non-empty set of vertices and $E$ is a collection of unordered pairs of distinct vertices, called edges. For an edge $e=\{u,v\}\in E$, we say that $e$ joins the vertices $u$ and $v$; equivalently, $u$ and $v$ are incident with $e$.
\end{definition}

\begin{definition}\label{def2}\cite{Bondy2008}
Given graph $G=(V,\,E)$, an edge $e\in E$ is called a loop if its two endpoints are the identical vertex, i.e., $e={v,\,v}$ for some $v\in V$.
\end{definition}

\begin{definition}\label{def3}\cite{Bondy2008}
 (Simple graph) A simple graph is an undirected graph with neither loops nor multiple edges.
\end{definition}

\begin{definition}\label{def4}\cite{Bondy2008} 
(Adjacency matrix) The adjacency matrix of a simple undirected graph $G$ with $n$ vertices $\{v_{j}:j=1,\,2,\,\ldots,\,n\}$ is an $n\times n$ symmetric binary matrix $A=(a_{ij})$, whose entries satisfy $a_{ij}=1$ if there exists an undirected edge connecting vertices $v_{i}$ and $v_{j}$, and $a_{ij}=0$ otherwise.
\end{definition}

\begin{definition}\label{def5}
(Two-colored adjacency matrix) The adjacency matrix of a simple undirected graph $G$ with $n$ vertices $\{v_{j}:j=1,\,2,\,\ldots,\,n\}$ is an $n\times n$ symmetric matrix $A=(a_{ij})$, whose entries satisfy $a_{ij}=1$ if there exists an edge with color $1$ connecting vertices $v_{i}$ and $v_{j}$, $a_{ij}=-1$ if there exists an edge with color $2$ connecting vertices $v_{i}$ and $v_{j}$, and $a_{ij}=0$ if there does not exist an edge connecting vertices $v_{i}$ and $v_{j}$. 
\end{definition}

\begin{definition}\label{def6}\cite{DiVincenzo2003} (Orthogonality graph)
	Let $\mathcal{H}=\bigotimes^{m}_{i=1}\mathcal{H}_{i}$ be an $m$-partite Hilbert space with dim $\mathcal{H}_{i}=d_{i}$. Let 
$S=\{\vert \psi_{j}\rangle\equiv\bigotimes^{m}_{i=1}\vert \varphi_{i,j}\rangle|j=1,2,...,n\}$ be an orthogonal product basis in $\mathcal{H}$. We represent $S$ as a graph $G = (V, E_{1}\cup E_{2}\cup \cdots \cup E_{m})$, where the set of edges $E_{i}$ have color $i$. The states $\vert \psi_{j} \rangle \in S$ are represented as the vertices $V$. There exists an edge $e$ of color $i$ between the vertices $v_{k}$ and $v_{l}$, i.e., $e \in E_{i}$, when states $\vert \phi_{k}\rangle$ and $\vert \phi_{l}\rangle$ are orthogonal on $\mathcal{H}_{i}$. Since all the states in the product basis are mutually orthogonal, every vertex is connected to all the other vertices by at least one edge of some color. The graph $G$ is called the orthogonality graph of the product basis.
\end{definition}

\begin{definition}\label{def7}\cite{XHWYJ} 
	(The vector of the numbers of pairwise orthogonality relations) A vector $\boldsymbol{(a, b)}$ is used to represent the numbers of pairwise orthogonality relations among the members of a set of bipartite OPSs, where $\boldsymbol{a}$ denotes the number of pairwise orthogonality relations among those states on the first subsystem while $\boldsymbol{b}$ denotes the number of pairwise orthogonality relations on the second subsystem. For simplicity, we call $\boldsymbol{(a, b)}$ the vector of the numbers of pairwise orthogonality relations.
\end{definition}

\begin{definition}\label{def8}\cite{Bondy2008}
	(Degree in a simple graph) Let $G=(V,\,E)$ be an undirected simple graph, and let $v\in V$ be a vertex of $G$. The degree of $v$,  denoted as $deg$($v$), is the number of edges of $G$ incident with $v$.
\end{definition}

\begin{definition}\label{def9}\cite{XHWYJ}
	Let $G = (V,\,E)$ be an orthogonality graph of a set of multipartite OPSs, and let $v\in V$ be a vertex of $G$. The degree of $i$-th party of $v$,  denoted as $deg_{i}$($v$), is the number of edges of $G$ incident with vertex $v$ in color $i$.
\end{definition}

\begin{definition}\label{def10}\cite{21}
	Suppose a measurement described by measurement operators $\{M_{m}|m=1,2,...,d\}$ is performed upon a quantum system in the state $\vert \psi \rangle$. We define $E_{m}\equiv M_{m}^{\dagger}M_{m}$, where $E_{m}$ is a positive operator such that $\sum_{m}E_{m}=I$. The operators $E_{m}$ are known as the positive operator-valued measure (POVM) elements associated with the measurement. The complete set $\{E_{m}\}$ is known as a POVM.
\end{definition}

\begin{lemma}\label{lm1} \cite{DiVincenzo2003}
	Let $S$ be a set of bipartite OPSs with four or fewer members in any dimension (that allows for this set of OPSs). The set $S$ is distinguishable by local incomplete von Neumann measurements and classical communication.
\end{lemma}

 Note that von Neumann measurement, namely the projection measurement, is a special POVM.

By Theorem 2 and Theorem 3 of  Ref. \cite{XHWYJ}, we have the following Lemma directly.

\begin{lemma}\label{lm2}\cite{XHWYJ}
  	Five bipartite OPSs, any two of which are orthogonal on only one subsystem, cannot be perfectly distinguished by LOCC or can be locally distinguished with a certain probability when 
all vertices $v_{j}$ in their orthogonality graph satisfy  $deg_{1}(v_{j})=deg_{2}(v_{j})=2$ for $j=1,\,2,\,\ldots,\,5$.  
    \end{lemma}

\section{\label{sec3}Local distinguishability of bipartite orthogonal product states}

In this section, we present three Theorems that can be used to simplify the proof of local distinguishability of six bipartite OPSs. For convenience, when discussing the local discrimination of bipartite quantum states, we always assume that Alice holds the first subsystem of a given state and Bob its second subsystem. It should be noted that any two product states discussed in this paper are orthogonal on only one subsystem.

\begin{theorem}\label{th1}
  For a set of bipartite OPSs, if one state is orthogonal to all other members on the same subsystem, then the state can be perfectly identified by LOCC, while the orthogonality relations among other members of the set can remain invariant. 
\end{theorem}

\begin{figure}[H]
	\centering
	\includegraphics[width=0.25\textwidth]{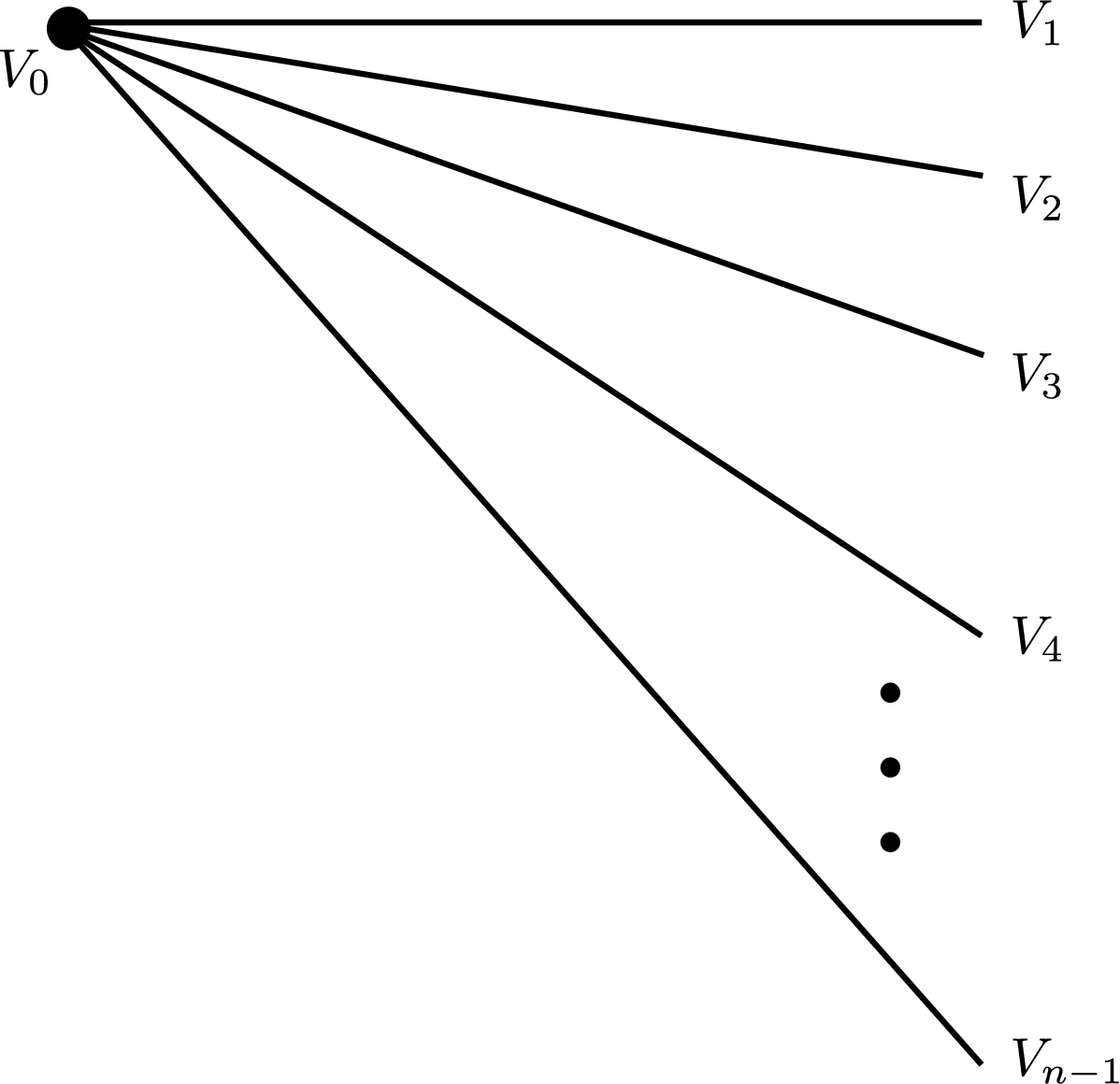}%
	\caption{Orthogonality relations between bipartite product state $V_{0}$ and bipartite product states $\{V_{i}\,|i=1,\,2,\,\ldots,\,n-1\}$ on 
    the first subsystem.\label{fig001}}
\end{figure}

\noindent\textbf{Proof.}
Without loss of generality, as shown in Fig.~\ref{fig001}, suppose that the bipartite product state $V_{0}$ is orthogonal to each of bipartite product states $\{V_{i}:\,1,\,2,\,\ldots,\,n-1\}$ on the first subsystem. For ease of presentation, we assume that the first subsystem of state $V_{i}$ is denoted by $|\alpha_{i}\rangle_{A}$, where $\langle\alpha_{i}|\alpha_{i}\rangle_{A}=1$ for $i=0,\,1,\,\ldots,\,n-1$.

Let Alice perform a POVM measurement with the operators $M_{1}=\vert \alpha_{0}\rangle_{A}\langle \alpha_{0}\vert$ and $M_{2}=I-\vert \alpha_{0}\rangle_{A}\langle \alpha_{0}\vert$ on the first subsystem of the measured state. Since $\vert \alpha_{0}\rangle_{A}$ is orthogonal to each of  $\{\vert \alpha_{i}\rangle_{A}:\, i=1,\,2,\,\ldots,\,n-1\}$, the measurement outcome 1 guarantees that the measured state must be state $V_{0}$. If the outcome is 2, the subsystem of the measured state is kept invariant since 
\begin{equation}
\nonumber
	\begin{aligned}
M_{2}|\alpha_{i}\rangle_{A}=(I-\vert \alpha_{0}\rangle_{A}\langle \alpha_{0}\vert)|\alpha_{i}\rangle_{A}=|\alpha_{i}\rangle_{A}
	\end{aligned}
    \end{equation}
for $i=1,\,2,\,\ldots,\,n-1.$ Consequently, the orthogonality relations among these $n-1$ states on Alice’s side are preserved, exactly as they were in the original ensemble. Because Bob does not perform any measurement, the orthogonality relations among these $n-1$ states on Bob’s side are likewise unaffected. This completes the proof. \hfill\rule{6pt}{6pt}

\begin{figure}[H]
	\centering
	\includegraphics[width=0.25\textwidth]{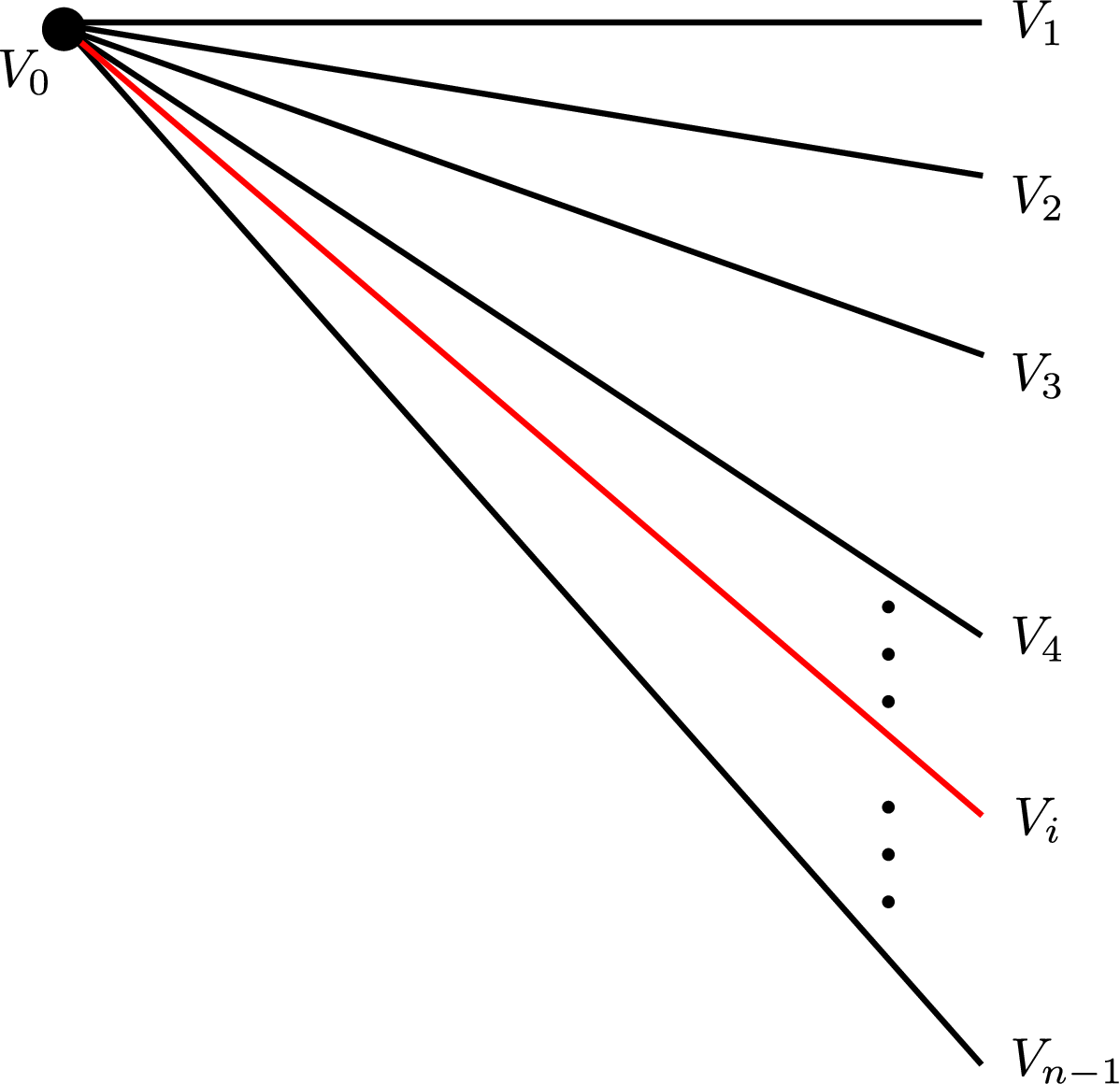}%
	\caption{Orthogonality relations between bipartite product state $V_{0}$ and bipartite product states $\{V_{i}\,|i=1,\,2,\,\ldots,\,n-1\}$.\label{fig002}}
\end{figure}

\begin{theorem}\label{th2}
For a set of bipartite OPSs, if one state is orthogonal to another state on one subsystem and orthogonal to all remaining states on the other subsystem, then this state can be perfectly discriminated via LOCC, while the orthogonality relations among all the other states can remain unchanged.
\end{theorem}

\noindent\textbf{Proof.} As shown in Fig.~\ref{fig002}, bipartite product state $V_{0}$ is orthogonal to state  $V_{i}$ on the second subsystem while  orthogonal to states  $V_{j}$ on the first subsystem for $1\leq j<i$ and $i< j\leq n-1$. Note that states in Fig.~\ref{fig002}  satisfy all the hypotheses of Theorem~\ref{th2}. It suffices to prove that these states satisfy the conclusion of Theorem~\ref{th2}. 

Suppose that the first and second subsystems of bipartite product state $V_{j}$ is denoted by $|\alpha_{j}\rangle_{A}$ and $|\beta_{j}\rangle_{B}$, respectively, for $0\leq j\leq n-1$. Let Alice perform a two-outcome POVM with measurement operators $M_{1}=\vert \alpha_{0}\rangle_{A}\langle \alpha_{0}\vert$ and $M_{2}=I-\vert \alpha_{0}\rangle_{A}\langle \alpha_{0}\vert$. (1) If the measurement outcome is 1, the measured state must be state $V_{0}$ or $V_{i}$. Because these two states are orthogonal on the second subsystem, Bob can continue with a local measurement to distinguish them. (2) If Alice’s measurement outcome is 2, then the measured state must be one of states $\{V_{j}:\,j=1,\,2,\,\ldots,\,n-1\}$ and the first subsystem that Alice holds has one of the forms as shown in Eq. (\ref{eq1}) after her measurement.

\begin{equation}
	\begin{aligned}
		\vert \alpha_{j}'\rangle_{A} &= M_{2}\vert \alpha_{j}\rangle_{A} = \vert \alpha_{j}\rangle_{A},\\
		\vert \alpha_{i}'\rangle_{A} &=M_{2}\vert \alpha_{i}\rangle_{A}= \vert \alpha_{i}\rangle_{A}-\langle \alpha_{0}\vert \alpha_{i}\rangle_{A} \vert \alpha_{0}\rangle_{A},\\
	\end{aligned}\label{eq1}
\end{equation}
where $1\leq j \leq n-1$ and $j \neq i$. We consider the orthogonality relations among post-measurement states $\{V^{\prime}_{j}:\,1\leq j\leq n-1\}$, where $V^{\prime}_{j}$ is the collapsed state of state $V_{j}$  after Alice's measurement. By Eq.~(\ref{eq1}), we have

\begin{equation}
\nonumber
	\begin{aligned}
    \langle\alpha_{j}'|\alpha_{k}'\rangle_{A} &=\langle\alpha_{j}|\alpha_{k}\rangle_{A},\\
    \langle\alpha_{j}'|\alpha_{i}'\rangle_{A} &=\langle\alpha_{j}|\alpha_{i}\rangle_{A},
    \end{aligned}
\end{equation}
where $1\leq j \leq n-1$, $j \neq i$, $1\leq k \leq n-1$ and $k \neq i$. This means that the orthogonality relations among states  $\{V_{j}:\,j=1,\,2,\,\ldots,\,n-1\}$ remain unchanged. The proof is complete. \hfill\rule{6pt}{6pt}

\begin{figure}[H]
	\centering
	\includegraphics[width=0.25\textwidth]{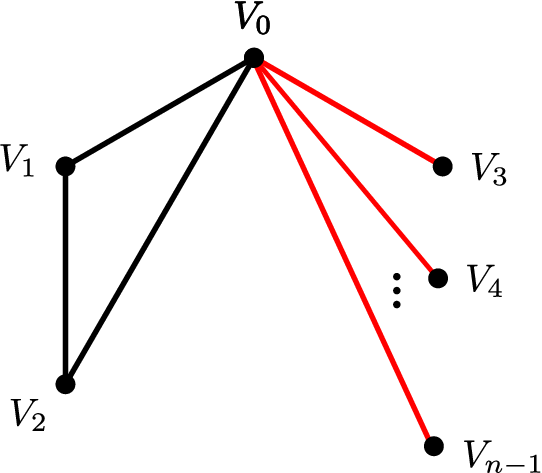}%
	\caption{Orthogonality relations among a set of bipartite OPSs $\{V_{i}\,|i=0,\,1,\,\ldots,\,n-1\}$.\label{fig003}}
\end{figure}

\begin{theorem}\label{th3}
For a set of bipartite OPSs, if three states are pairwise orthogonal on a certain subsystem, and one of the three states is orthogonal on the other subsystem to all states in the set except the three, then that state can be perfectly distinguished by LOCC while the orthogonality relations of all remaining states can stay intact. 
\end{theorem}

\noindent\textbf{Proof.} As shown in Fig.~\ref{fig003}, bipartite OPSs $\{V_{j}:j=0,\,1,\,2\}$ are pairwise orthogonal on the first subsystem; state $V_{0}$ is orthogonal to each of bipartite product states $\{V_{3},\,V_{4},\,\ldots,\,V_{n-1}\}$ on the second subsystem. That is, states in Fig.~\ref{fig003} satisfy the hypotheses of Theorem~\ref{th3}. To streamline the presentation of this proof, we carry out the derivation based on the structure depicted in Fig. \ref{fig003}. 
Suppose that the first and second subsystems of bipartite product state $V_{j}$ is denoted by $|\alpha_{j}\rangle_{A}$ and $|\beta_{j}\rangle_{B}$, respectively, where 
$\langle\alpha_{j}|\alpha_{j}\rangle_{A}=\langle\beta_{j}|\beta_{j}\rangle_{B}=1$ for $0\leq j\leq n-1$.

Initially, let Bob perform a POVM with the operators $M_{1}=\vert \beta_{0}\rangle_{B}\langle \beta_{0} \vert$ and $M_{2}=I-\vert \beta_{0}\rangle_{B}\langle \beta_{0} \vert$. (1) If the measurement outcome is 1, the measured state must be one of the states $\{V_{j}:j=0,\,1,\,2\}$. Since these three states are mutually orthogonal on Alice’s side, a subsequent measurement by Alice will uniquely identify which of the three states the measured state is. (2) If the measurement outcome is 2, the measured state must be one of the states $\{V_{j}:j=1,\,2,\,\ldots,\,n-1\}$. The second subsystem that Bob holds has one of the forms as shown in Eq. (\ref{eq2}) after his measurement. 

\begin{equation}
	\begin{aligned}
		\vert \beta_{k}'\rangle_{B} &= M_{2}\vert \beta_{k}\rangle_{B}=\vert \beta_{k}\rangle_{B} -\langle\beta_{0}\vert\beta_{k}\rangle\vert \beta_{0}\rangle_{B},\\
        \vert \beta_{j}'\rangle_{B} &=M_{2}\vert \beta_{j}\rangle_{B}= \vert \beta_{j}\rangle_{B},\\
    \end{aligned}\label{eq2}
\end{equation}
where $k=1,\,2$, and $j=3,\,4,\,\ldots,\,n-1$.

We consider the orthogonality relations among post-measurement states $\{V^{\prime}_{j}:\,1\leq j\leq n-1\}$, where $V^{\prime}_{j}$ is the collapsed state of state $V_{j}$  after Bob's measurement. By Eq.~(\ref{eq2}), we have

\begin{equation}
\nonumber
	\begin{aligned}
    \langle\beta_{j}'|\beta_{k}'\rangle_{A} &=\langle\beta_{j}|\beta_{k}\rangle_{A},\\
    \langle\beta_{j}'|\beta_{t}'\rangle_{A} &=\langle\beta_{j}|\beta_{t}\rangle_{A},\\
     \end{aligned}
\end{equation}
where $3\leq j \leq n-1$, $3\leq k \leq n-1$, and $t=1,\,2$. This means that the orthogonality relations among states  $\{V_{j}:\,j=1,\,2,\,\ldots,\,n-1\}$ remain unchanged. The proof is complete. \hfill\rule{6pt}{6pt}

\section{Local distinguishability of six orthogonal product states}\label{sec4}

In this section, we discuss the local distinguishability of six bipartite OPSs, where any two of them are orthogonal on only one subsystem. In fact, a set of six bipartite OPSs has 15 pairwise orthogonality relations if any two of these states are orthogonal on only one subsystem. On the other hand, it is evident that six bipartite OPSs with the vector of pairwise orthogonality relations $(a,\,b)$ have identical local distinguishability to those with vector $(b, \,a)$. Thus, we need to consider only one case. 

In terms of structural features, all feasible sets of six bipartite OPSs are partitioned into eight categories  by the vectors of the numbers of pairwise orthogonality relations, i.e., (15, 0), (14, 1), (13, 2), (12, 3), (11, 4), (10, 5), (9, 6), (8, 7). For ease of understanding, we characterize the structure of each set of six bipartite OPSs belonging to every category via orthogonality graphs. For these categories, there exist a total of 78 distinct configurations, among which only five cases correspond to sets of six bipartite OPSs that are not locally distinguishable via LOCC. We enumerate all feasible cases in Table~\ref{tab1} according to the above classification. Note that we have omitted configurations that coincide with listed ones either under party interchange or after state relabeling, as these cases will follow the same line of reasoning.
 \begin{table}[H]
	\centering
	\caption{All feasible configurations of six bipartite OPSs}\label{tab1}
	\begin{ruledtabular}
		\begin{tabular}{c c c c}
			Categories &   Cases can be      &  Cases for  & Total case \\
                     &     perfectly       &   others    &  numbers in \\
                     & LOCC distinguished  &             & these categories\\
			\noalign{\smallskip}\hline\noalign{\smallskip}
			(15,\,0)     & 1 & 0 & 1\\
            (14,\,1)     & 1 & 0 & 1\\
            (13,\,2)     & 2 & 0 & 2\\
            (12,\,3)     & 5 & 0 & 5\\
            (11,\,4)     & 9 & 0 & 9\\
            (10,\,5)     & 14 & 1 & 15\\
            (9,\,6)     & 19 & 2 & 21\\
            (8,\,7)     & 22 & 2 & 24\\
		\end{tabular}
	\end{ruledtabular}
\end{table}

 Subsequently, we present the orthogonality graphs for all feasible configurations of six bipartite OPSs from Fig. \ref{fig004} to Fig. \ref{fig8}. In these orthogonality graphs, a black edge indicates that two OPSs are orthogonal only on the first subsystem, while a red edge denotes that they are orthogonal only on the second subsystem. Similarly, all graphs that coincide with the existing orthogonality graphs via edge color swapping or index relabeling are omitted in this paper. In addition to depicting the orthogonality graphs of the six bipartite OPSs, these figures also illustrate the discrimination processes for the configurations that can be perfectly distinguished via LOCC. The circles of different colors in each figure represent the distinct measurement strategies employed in the corresponding LOCC protocols. In these figures, a yellow highlight on the sub-index to the right of each graph marks a case that cannot be perfectly distinguished by LOCC. There are five such cases in total.
    
\begin{figure}[H]
	\centering
	\includegraphics[width=0.48\textwidth]{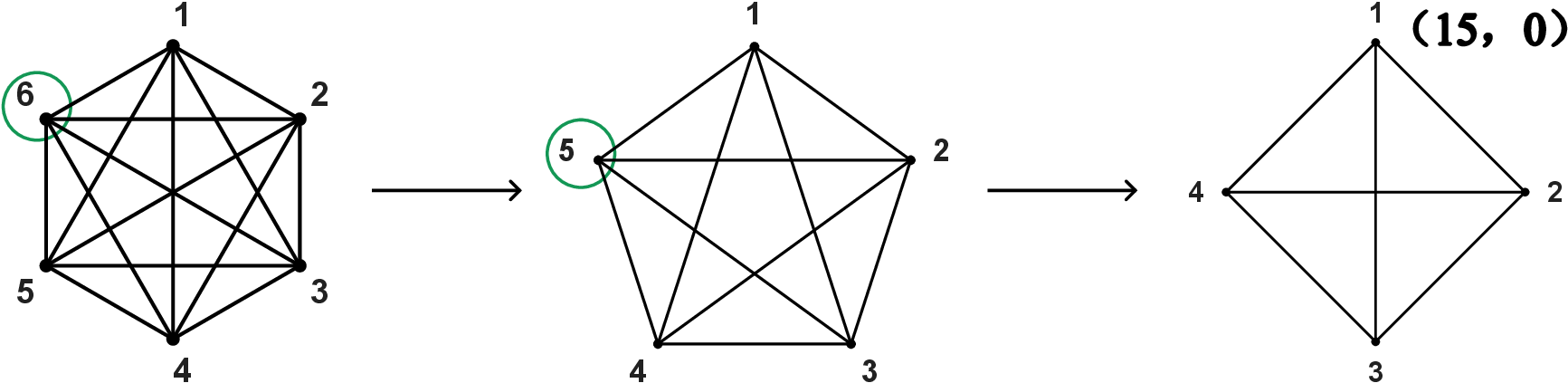}%
	\caption{Feasible orthogonality graph and discrimination process diagram for category (15, 0).\label{fig004}}
\end{figure}

\begin{figure}[H]
	\centering
	\includegraphics[width=0.48\textwidth]{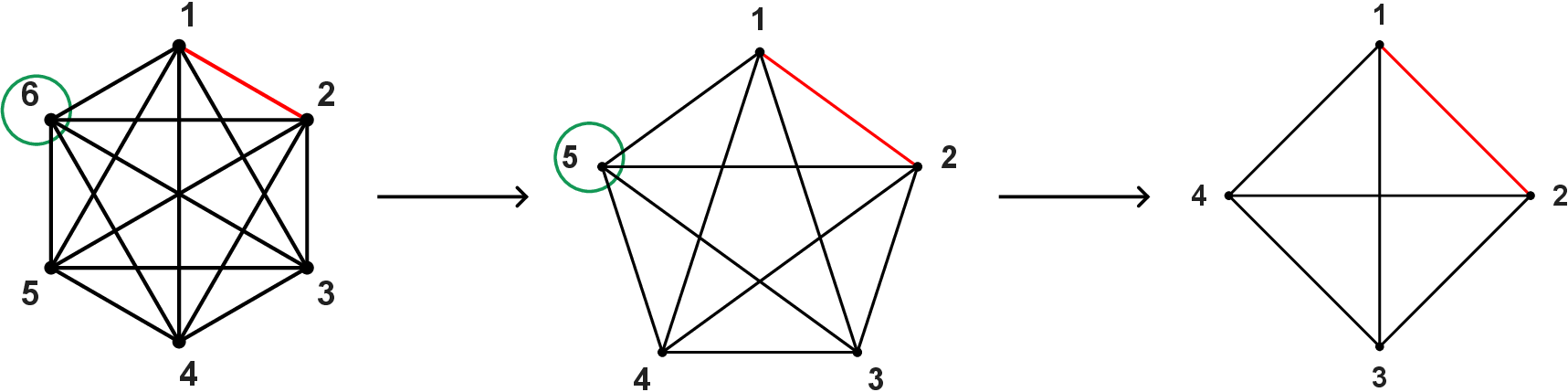}%
	\caption{Feasible orthogonality graph and discrimination process diagram for category (14, 1).\label{fig2}}
\end{figure}

\begin{figure}[H]
	\centering
	\includegraphics[width=0.48\textwidth]{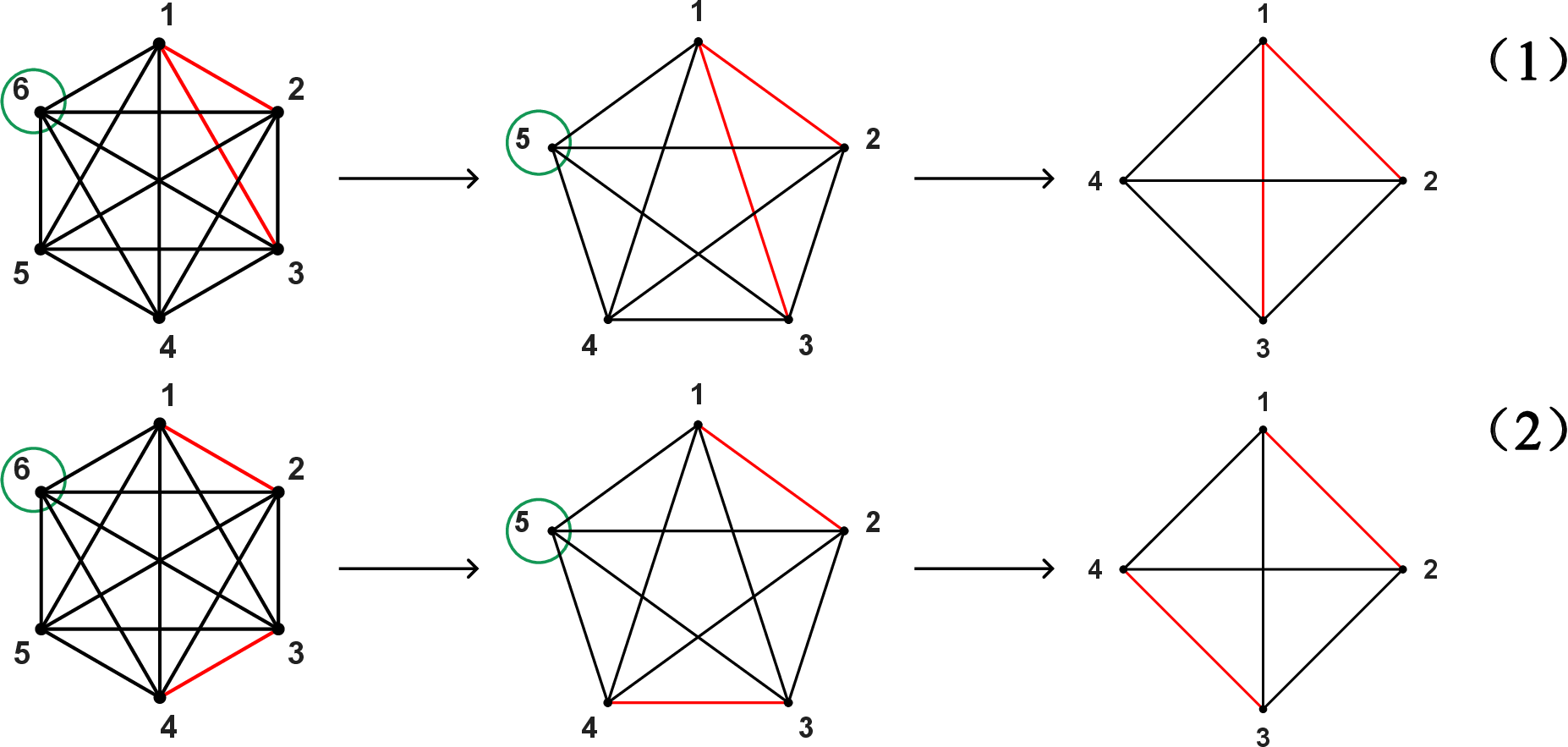}%
	\caption{Feasible orthogonality graphs and discrimination process diagrams for category (13, 2).\label{fig3}}
\end{figure}

\begin{figure}[H]
	\centering
	\includegraphics[width=0.48\textwidth]{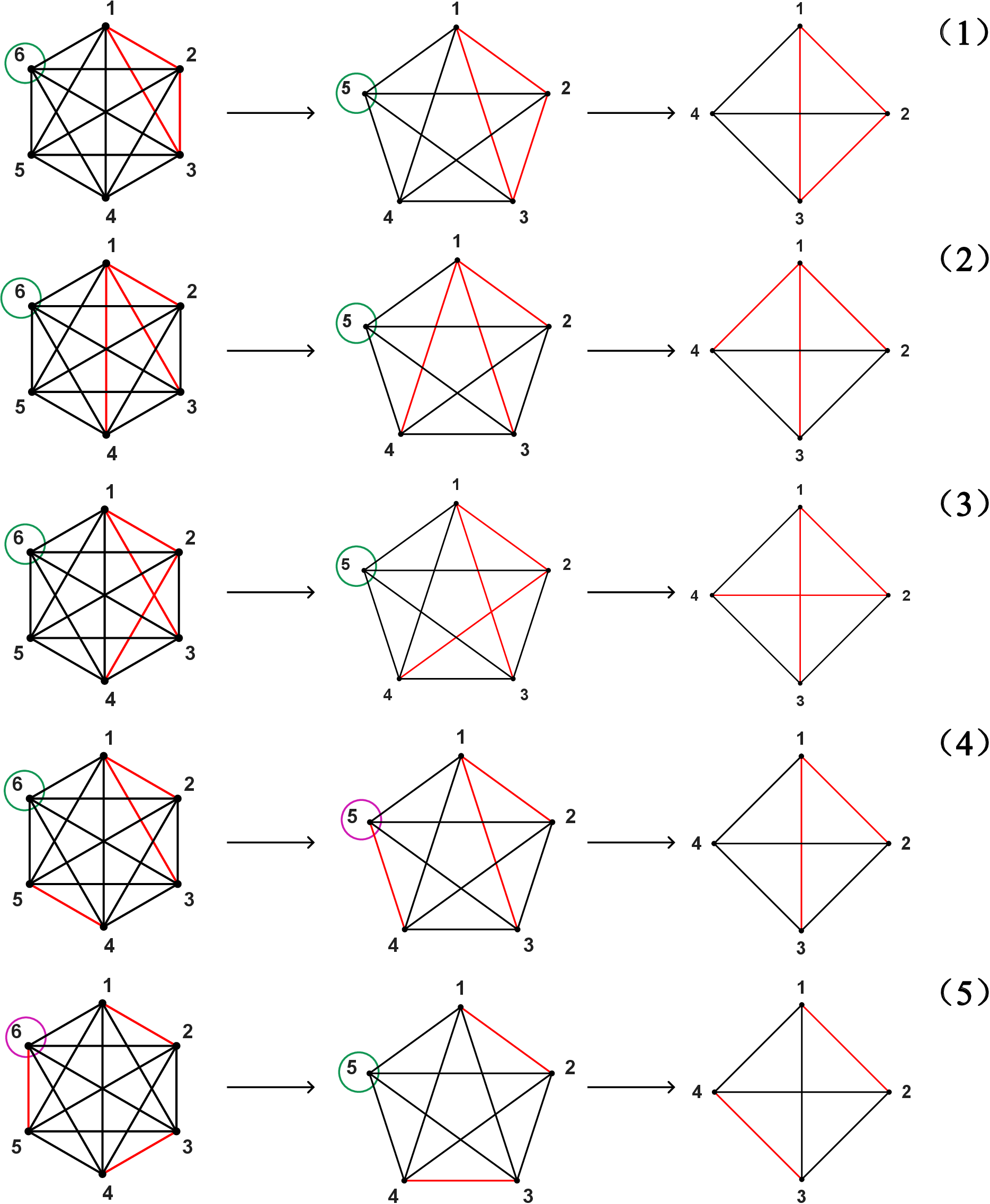}%
	\caption{Feasible orthogonality graphs and discrimination process diagrams for category (12, 3).\label{fig4}}
\end{figure}

\begin{figure*}[t]
	\centering
	\includegraphics[width=\textwidth]{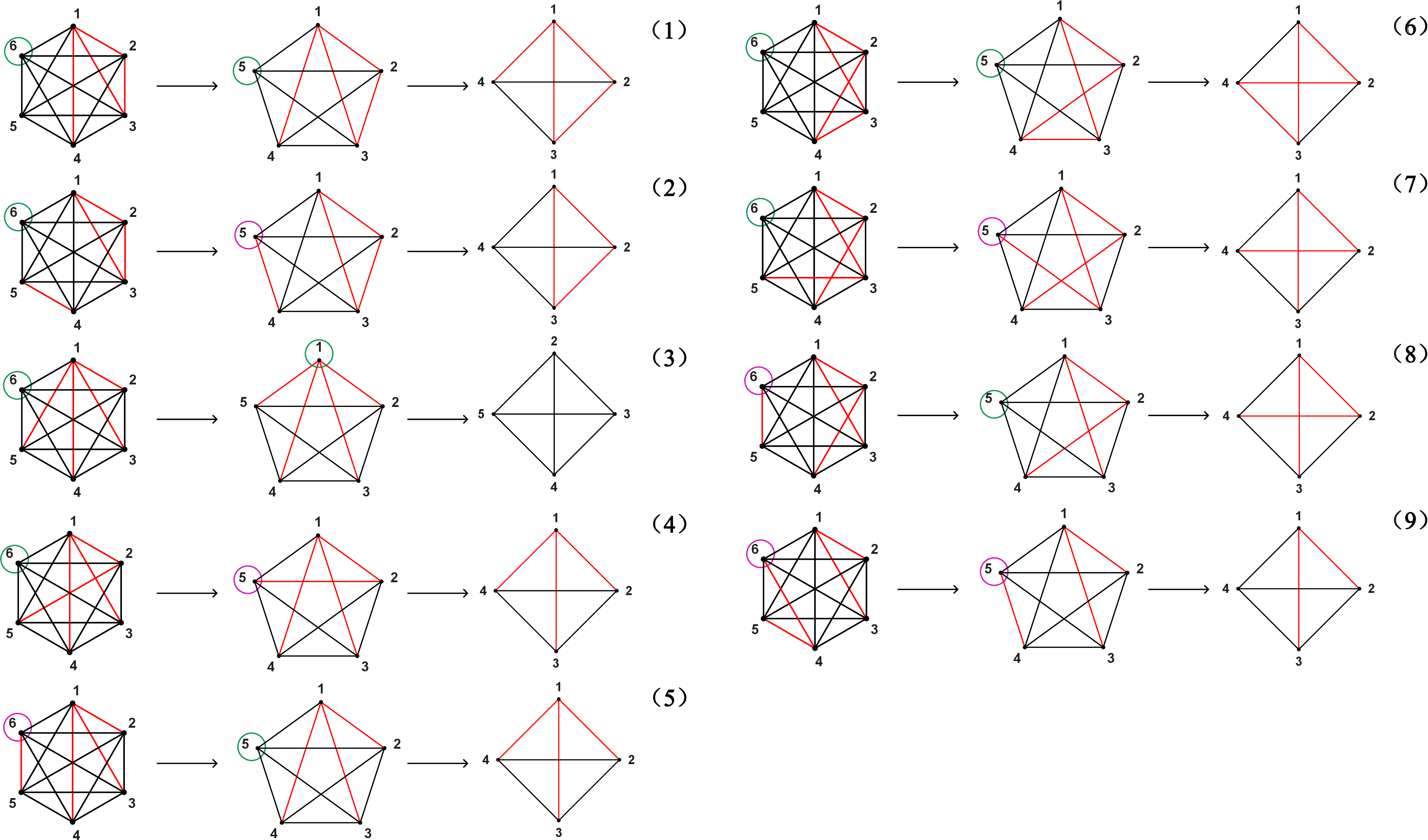}%
	\caption{Feasible orthogonality graphs and discrimination process diagrams for category (11, 4).\label{fig5}}
\end{figure*}

\begin{figure*}[t]
	\centering
	\includegraphics[width=\textwidth]{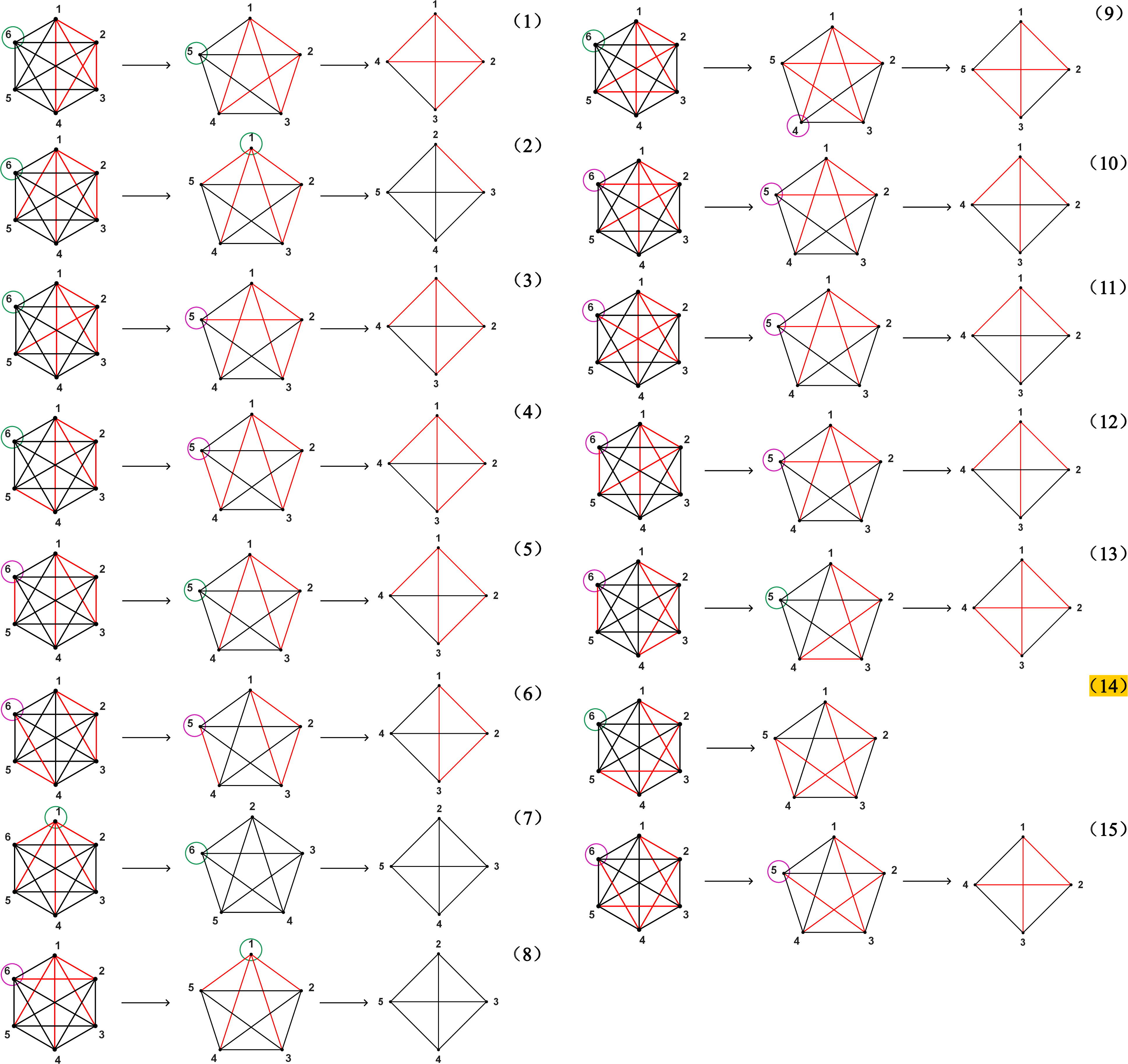}%
	\caption{Feasible orthogonality graphs and discrimination process diagrams for category (10, 5).\label{fig6}}
\end{figure*}

\begin{figure*}[t]
	\centering
	\includegraphics[width=\textwidth]{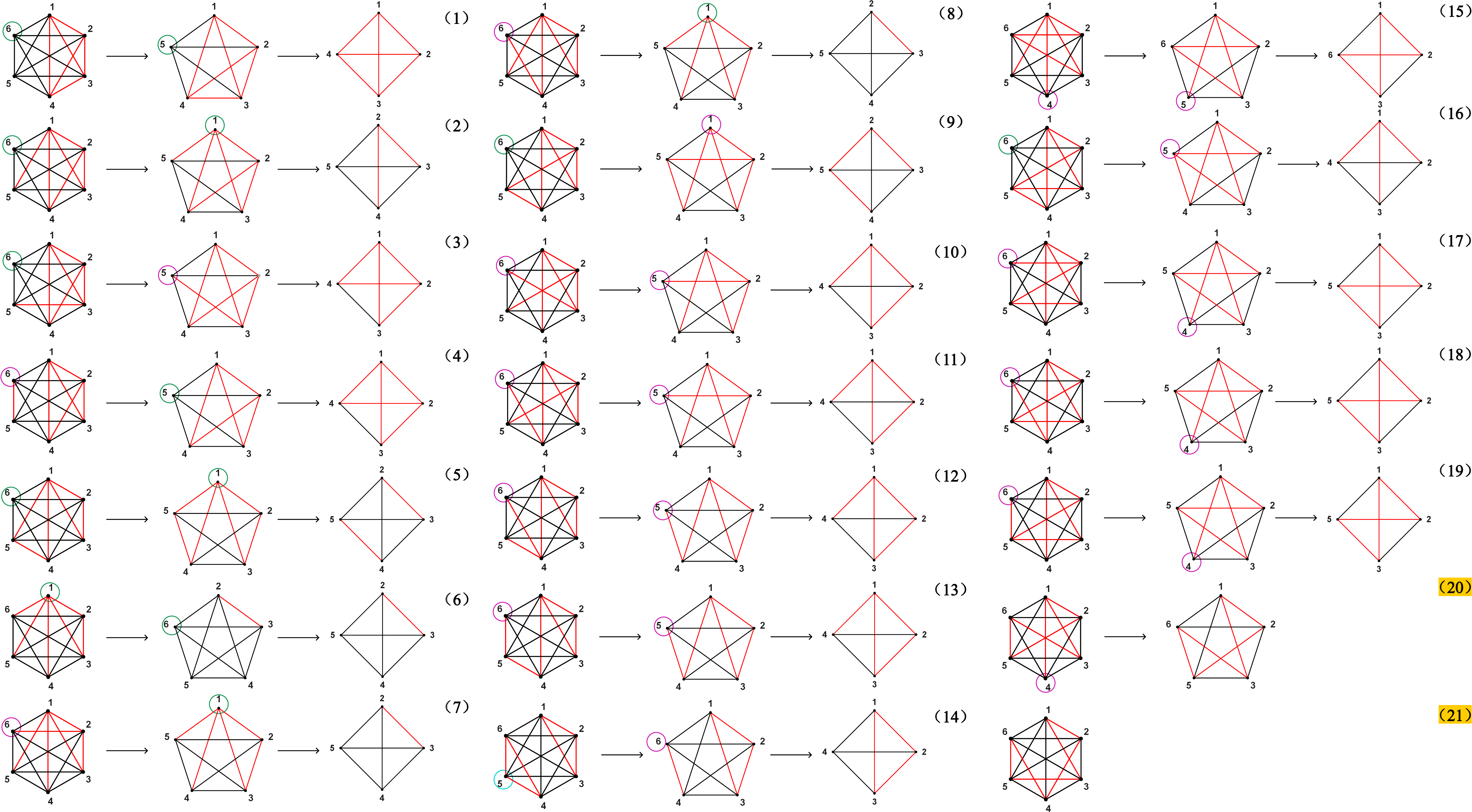}%
	\caption{Feasible orthogonality graphs and discrimination process diagrams for category (9, 6).\label{fig7}}
\end{figure*}

\begin{figure*}[t]
	\centering
	\includegraphics[width=\textwidth]{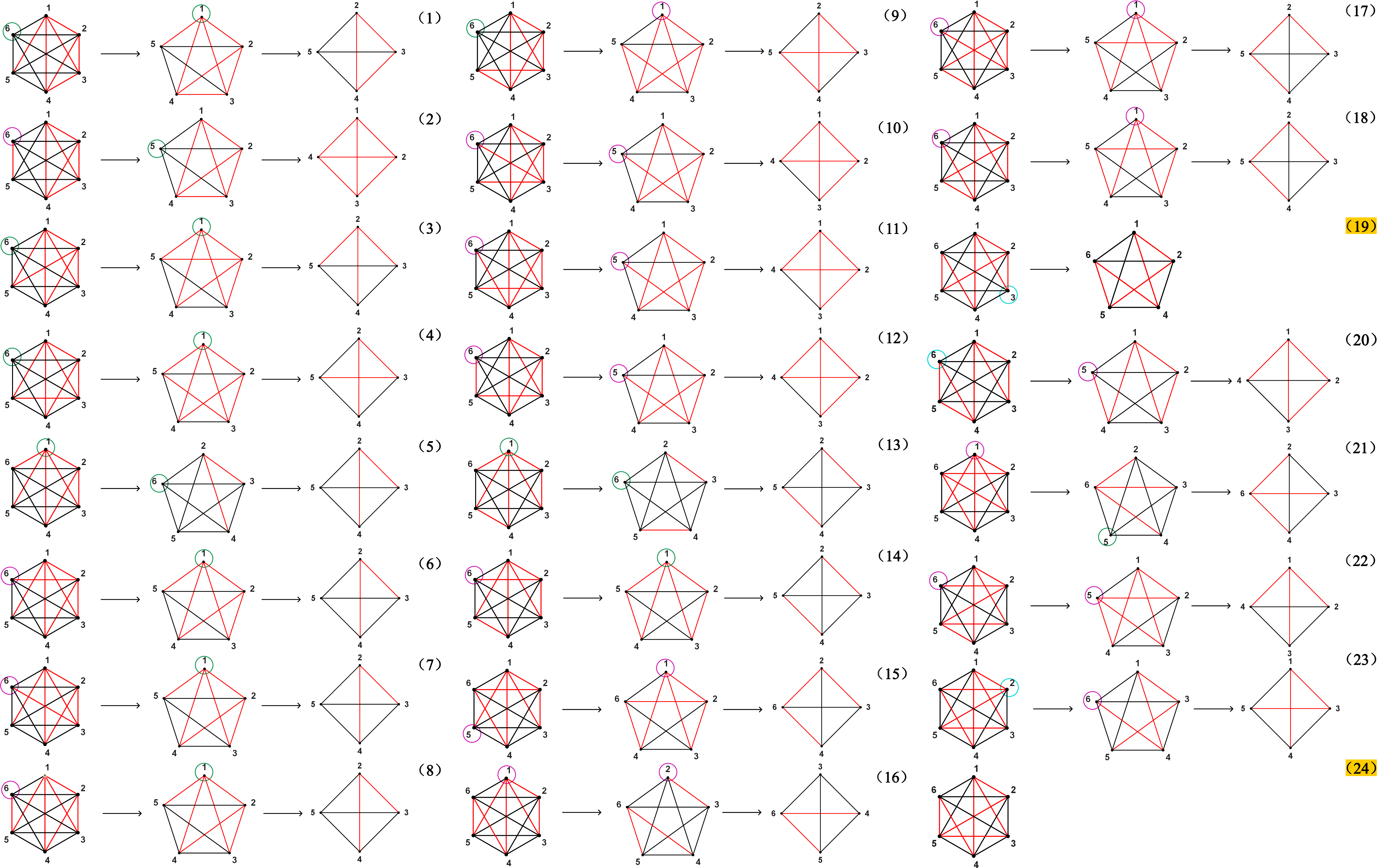}%
	\caption{Feasible orthogonality graphs and discrimination process diagrams for category (8, 7).\label{fig8}}
\end{figure*}

\begin{theorem}\label{th4}
	Six bipartite OPSs, any two of which are orthogonal on only one subsystem, can be perfectly distinguished via LOCC except that their two-colored adjacency matrix corresponding to their orthogonality graph after relabeling the vertices takes one of the following five forms

\begin{equation}\label{eq3}   
   \begin{pmatrix}0&-1&-1&1&1&1\\-1&0&1&-1&1&1\\-1&1&0&1&-1&1\\1&-1&1&0&-1&1\\1&1&-1&-1&0&1\\1&1&1&1&1&0\end{pmatrix}, \\
\end{equation}
\begin{equation}\label{eq4}
    \begin{pmatrix}0&-1&-1&-1&1&1\\-1&0&1&1&-1&1\\-1&1&0&1&1&-1\\-1&1&1&0&1&1\\1&-1&1&1&0&-1\\1&1&-1&1&-1&0\end{pmatrix}, \\\
\end{equation}
\begin{equation}\label{eq5}
   \begin{pmatrix}0&-1&-1&-1&1&1\\-1&0&-1&1&-1&1\\-1&-1&0&1&1&1\\-1&1&1&0&1&-1\\1&-1&1&1&0&-1\\1&1&1&-1&-1&0\end{pmatrix}, \\
\end{equation}
\begin{equation}\label{eq6}
   \begin{pmatrix}0&-1&-1&1&1&1\\-1&0&1&-1&1&1\\-1&1&0&1&-1&1\\1&-1&1&0&1&-1\\1&1&-1&1&0&-1\\1&1&1&-1&-1&0\end{pmatrix}, \\
\end{equation}
\begin{equation}\label{eq7}
 \begin{pmatrix}0&-1&-1&-1&1&1\\-1&0&1&1&-1&1\\-1&1&0&1&-1&1\\-1&1&1&0&1&-1\\1&-1&-1&1&0&-1\\1&1&1&-1&-1&0\end{pmatrix}. 
\end{equation}
\end{theorem}
\begin{proof} As stated earlier, all feasible configurations of six bipartite OPSs are presented in Figs. 4–11. We assume that the first subsystem of state $j$ is $|\alpha_{j}\rangle_{A}$ and the second subsystem $|\beta_{j}\rangle_{B}$ for $j=1,\,2,\,\ldots,\,6$ in each case. We present the proof by cases below.

\begin{center} 1. Category (15, 0)\end{center}

As shown in Fig. \ref{fig004}, states $\{1,\,2,\,\ldots,6\}$ are pairwise orthogonal on only the first subsystem in this category. Therefore, each of states $\{1,\,2,\,\ldots,6\}$ can be perfectly distinguished by Alice with the measurement operators $\{|j\rangle\langle j|:j=1,\,2,\,\ldots,\,6\}$.

\begin{center} 2. Category (14, 1)\end{center}

As shown in Fig. \ref{fig2}, state 6 is orthogonal to each of states $\{1,$ $\,2,$ $\,\ldots,$ $\,5\}$ on only the first subsystem. By Theorem \ref{th1}, state 6 can be perfectly distinguished by LOCC and the orthogonality relations among states $\{1,$ $\,2,$ $\,\ldots,$ $\,5\}$ can remain invariant. Thus state 5 is orthogonal to each of states $\{1,$ $\,2,$ $\,3,$ $\,4\}$ on only the first subsystem. By further utilizing Theorem 1, state 5 can be locally distinguished and the orthogonality relations among states $\{1,$ $\,2,$ $\,3,$ $\,4\}$ can remain invariant. By Lemma \ref{lm1}, states $\{1,$ $\,2,$ $\,3,$ $\,4\}$ can be perfectly distinguished by LOCC.

\begin{center} 3. Category (13, 2)\end{center}

As shown in Fig. \ref{fig3}, there exist two different feasible graphs for this category. Since state 6 is orthogonal to each of states $\{1,$ $\,2,$ $\,\ldots,$ $\,5\}$ on only the first subsystem, it can be locally distinguished and the orthogonality relations among states $\{1,$ $\,2,$ $\,\ldots,$ $\,5\}$ can remain invariant by Theorem \ref{th1}. Thus state 5 is orthogonal to each of states $\{1,$ $\,2,$ $\,3,$ $\,4\}$ on only the first subsystem. By further utilizing Theorem 1, state 5 can be perfectly distinguished by LOCC and the orthogonality relations among states $\{1,$ $\,2,$ $\,3,$ $\,4\}$ can remain invariant. By Lemma \ref{lm1}, states $\{1,$ $\,2,$ $\,3,$ $\,4\}$ can be perfectly distinguished by LOCC. 

\begin{center} 4. Category (12, 3)\end{center}

As shown in Fig. \ref{fig4}, there exist five different feasible graphs for this category. 

For each of cases (1)-(3), the local discrimination process of states $\{1,\,2,\,3,4,\,5,\,6\}$ is similar to that of  Category (13, 2).

For case (4), state 6 is orthogonal to each of states $\{1,$ $\,2,$ $\,\ldots,$ $\,5\}$ on only the first subsystem, it can be locally distinguished and the orthogonality relations among states $\{1,$ $\,2,$ $\,\ldots,$ $\,5\}$ can remain invariant by Theorem \ref{th1}. Since state 5 is orthogonal to states $\{1,\,2,\,3\}$on only the first subsystem and is orthogonal to state 4 on only the second subsystem, it can be exactly identified by LOCC, and the orthogonality relations among states $\{1,$ $\,2,$ $\,3,$ $\,4\}$ can remain invariant by Theorem \ref{th2}. By Lemma \ref{lm1}, states $\{1,$ $\,2,$ $\,3,$ $\,4\}$ can be perfectly distinguished by LOCC.  

For case (5), state 6 is orthogonal to states \{1, 2, 3, 4\} on only the first subsystem and is orthogonal to state 5 on only the second subsystem. By Theorem \ref{th2},  state 6 can be perfectly distinguished by LOCC and the orthogonality relations among states $\{1,$ $\,2,$ $\,3,$ $\,4,\,5\}$ can remain invariant. By Theorem \ref{th1} and Lemma \ref{lm1}, states $\{1,$ $\,2,$ $\,3,$ $\,4,\,5\}$ can be perfectly distinguished by LOCC. 

\begin{center} 5. Category (11, 4)\end{center}

As shown in Fig. \ref{fig5}, there exist nine different feasible graphs for this category. 

For each of cases (1) and (6), state 6 is orthogonal to each of states $\{1,$ $\,2,$ $\,3,$ $\,4,\,5\}$ on only the first subsystem. By Theorem \ref{th1}, state 6 can be perfectly distinguished by Alice and the orthogonality relations among states $\{1,$ $\,2,$ $\,3,$ $\,4,\,5\}$ can remain invariant. By Theorem \ref{th1} and Lemma \ref{lm1},  states $\{1,$ $\,2,$ $\,3,$ $\,4,\,5\}$ can be perfectly distinguished by LOCC.

For each of cases (2), (4) and (7), by Theorem \ref{th1}, state 6 can be perfectly distinguished by LOCC. Then state 5 can be exactly identified by LOCC by Theorem \ref{th2}. And states 
$\{1,$ $\,2,$ $\,3,$ $\,4\}$ can be locally distinguished by Lemma \ref{lm1}.

For each of cases (5) and (8), state 6 is orthogonal to each of states $\{1,$ $\,2,$ $\,3,$ $\,4\}$ on only the first subsystem and is orthogonal to state 5 on only the second subsystem. 
By Theorem \ref{th2}, state 6 can be perfectly distinguished by LOCC and the orthogonality relations among states $\{1,$ $\,2,$ $\,3,$ $\,4,\,5\}$ can remain invariant. By Theorem \ref{th1} and Lemma \ref{lm1}, states $\{1,$ $\,2,$ $\,3,$ $\,4,\,5\}$ can be perfectly distinguished by LOCC.

For case (3), state 6 is orthogonal to each of states $\{1,$ $\,2,$ $\,3,$ $\,4,\,5\}$ on only the first subsystem. By Theorem \ref{th1},  state 6 can be perfectly distinguished
by LOCC and the orthogonality relations among states \{1, 2, 3, 4, 5\} can remain invariant. By Theorem \ref{th1} and Lemma \ref{lm1}, states $\{1,$ $\,2,$ $\,3,$ $\,4,\,5\}$ can be perfectly distinguished by LOCC since state 1 is orthogonal to each of states \{2, 3, 4, 5\}.

\begin{center} 6. Category (10, 5)\end{center}

As shown in Fig. \ref{fig6}, there exist 15 different feasible graphs for this category. 

For case (14), two-colored adjacency matrix of six bipartite OPSs is as shown in Eq. (3) and their local distinguishability will be discussed in next section.

For each of other cases, six bipartite OPSs can be perfectly distinguished by LOCC. In the subfigure corresponding to each case, state circled in green indicates that it can be perfectly distinguished by LOCC and the orthogonality relations among other remaining states can remain invariant by Theorem \ref{th1}. Similarly,  in the subfigure corresponding to each case, state circled in pink indicates that it can be perfectly distinguished by LOCC and the orthogonality relations among other remaining states can remain invariant by Theorem \ref{th2}. In these cases, one state can be identified or eliminated after each local measurement. After two local measurements, only four states remain and their original orthogonality relations stay invariant. According to Lemma \ref{lm1}, the remaining four states are locally distinguishable.

\begin{center} 7. Category (9, 6)\end{center}

As shown in Fig. \ref{fig7}, there exist twenty-one different feasible graphs for this category. 

In the subfigure corresponding to one of cases (1)-(13) and cases (15)-(19), state circled in green indicates that it can be perfectly distinguished by LOCC and the orthogonality relations among other remaining states can remain invariant by Theorem \ref{th1}. Similarly, state circled in pink indicates that it can be perfectly distinguished by LOCC and the orthogonality relations among other remaining states can remain invariant by Theorem \ref{th2}. In these cases, one state can be identified or eliminated after each local measurement. After two local measurements, only four states remain and their original orthogonality relations stay invariant. According to Lemma \ref{lm1}, the remaining four states are locally distinguishable.

For case (14), in its orthogonality graph, state 5 circled in light blue is orthogonal to states 4 and 6 on only the second subsystem while is orthogonal to states 1, 2, and 3 on only the first subsystem. By Theorem \ref{th3}, state 5 can be exactly identified by LOCC and the orthogonality relations among states $\{1,\,2,\,3,\,4,\,6\}$ can remain invariant.  By Theorem \ref{th2} and Lemma \ref{lm1}, states  $\{1,\,2,\,3,\,4,\,6\}$ can be perfectly distinguished by LOCC. 

Thus, all cases can be perfectly distinguished by LOCC except for cases (20) and (21).

\begin{center} 8. Category (8, 7)\end{center}

As shown in Fig. \ref{fig8}, there exist twenty-four different feasible graphs for this category. 

In the subfigure corresponding to one of cases (1)-(18) and cases (20)-(23), state circled in green indicates that it can be perfectly distinguished by LOCC and the orthogonality relations among other remaining states can remain invariant by Theorem \ref{th1}, state circled in pink indicates that it can be perfectly distinguished by LOCC and the orthogonality relations among other remaining states can remain invariant by Theorem \ref{th2}, and state circled in light blue indicates that it can be perfectly distinguished by LOCC and the orthogonality relations among other remaining states can remain invariant by Theorem \ref{th3}. In these cases, one state can be identified or eliminated after each local measurement. After two local measurements, only four states remain and their original orthogonality relations stay invariant. According to Lemma 1, the remaining four states are locally distinguishable. Thus, all cases can be perfectly distinguished by LOCC except for cases (19) and (24). 

In Summary, with the exception of five special cases, namely case (14) in Fig. \ref{fig6}, cases (20) and (21) in Fig. \ref{fig7}, as well as cases (19) and (24) in Fig. \ref{fig8}, six bipartite OPSs corresponding to all feasible orthogonality graphs are locally distinguishable by LOCC. Note that the two-colored adjacency matrices of these five cases are as shown in Eqs. (\ref{eq3})-(\ref{eq7}). This completes the proof.\end{proof}
    	
\section{Analysis of five special cases}\label{sec5}  

In this section we give a detailed analysis of those five special cases, i.e., case (14) in Fig. \ref{fig6}, cases (20) and (21) in Fig. \ref{fig7}, as well as cases (19) and (24) in Fig. \ref{fig8}. 

\begin{center} 1. Case (14) in Fig. \ref{fig6} \end{center}

As shown in Fig. \ref{fig6}, for case (14), state 6 is orthogonal to each of states $\{1,\,2,\,3,\,4,\,5\}$ on only the first subsystem. By Theorem \ref{th1}, state 6 can be exactly identified by LOCC and the orthogonality relations among states $\{1,\,2,\,3,\,4,\,5\}$ remain invariant. Thus, the orthogonality graph of states $\{1,\,2,\,3,\,4,\,5\}$ is shown as the right subgraph of case (14) in Fig. \ref{fig6}. It is easy to see that $deg_{1}(j)=deg_{2}(j)=2$ for $j=1,\,2,\,\ldots,\,5$. By Lemma \ref{lm2}, states $\{1,\,2,\,3,\,4,\,5\}$ cannot be perfectly distinguished by LOCC or can be locally distinguished with a certain probability.

\begin{center} 2. Case (20) in Fig. \ref{fig7} \end{center}

As shown in Fig. \ref{fig7}, for case (20), state 4 is orthogonal to each of states $\{2,\,3,\,5,\,6\}$ on only the first subsystem and is orthogonal to state 1 on only the second subsystem. By Theorem \ref{th2}, state 4 can be exactly identified by LOCC and the orthogonality relations among states $\{1,\,2,\,3,\,5,\,6\}$ remain invariant. Thus, the orthogonality graph of states $\{1,\,2,\,3,\,5,\,6\}$ is shown as the right subgraph of case (20) in Fig. \ref{fig7}. It is easy to see that $deg_{1}(j)=deg_{2}(j)=2$ for $j=1,\,2,\,3,\,5,\,6$. By Lemma \ref{lm2}, states $\{1,\,2,\,3,\,5,\,6\}$ cannot be perfectly distinguished by LOCC or can be locally distinguished with a certain probability.

\begin{center} 3. Case (19) in Fig. \ref{fig8} \end{center}

As shown in Fig. \ref{fig8}, for case (19), state 3 is orthogonal to each of states $\{4,\,5,\,6\}$ on only the first subsystem and states $\{1,\,2,\,3\}$ are pairwise orthogonal on only the second subsystem. By Theorem \ref{th3}, state 3 can be exactly distinguished by LOCC and the orthogonality relations among states $\{1,\,2,\,4,\,5,\,6\}$ remain invariant. Thus, the orthogonality graph of states $\{1,\,2,\,4,\,5,\,6\}$ is shown as the right subgraph of case (19) in Fig. \ref{fig8}. It is easy to see that $deg_{1}(j)=deg_{2}(j)=2$ for $j=1,\,2,\,4,\,5,\,6$. By Lemma \ref{lm2}, states $\{1,\,2,\,4,\,5,\,6\}$ cannot be perfectly distinguished by LOCC or can be locally distinguished with a certain probability.

\begin{center} 4. Case (21) in Fig. \ref{fig7} \end{center}

In Eq. (\ref{eq8}), we present a set of six bipartite OPSs that correspond to the orthogonality graph of case (21) in Fig. \ref{fig7}. The set in Eq. (\ref{eq8}) is locally indistinguishable by LOCC. We now give the proofs of local indistinguishability of states in Eq. (\ref{eq8}).

\begin{equation}
	\begin{aligned}
		\vert \phi_{1}\rangle &=  \vert 0\rangle_{1}\vert 0\rangle_{2},\\
		\vert \phi_{2}\rangle &= \frac{1}{\sqrt{3}} (\vert 0\rangle+\vert 1\rangle+\vert 3\rangle)_{1}\vert 1\rangle_{2},\\
		\vert \phi_{3}\rangle &=\frac{1}{\sqrt{6}} (\vert 0\rangle - \vert 2\rangle-\vert 3\rangle)_{1}(\vert 1\rangle + \vert 2\rangle)_{2},\\
		\vert \phi_{4}\rangle &=\frac{1}{\sqrt{2}}\vert 1\rangle_{1} (\vert 0\rangle + \vert 2\rangle)_{2},\\
		\vert \phi_{5}\rangle &= \frac{1}{\sqrt{6}}\vert 2\rangle_{1}(2\vert 0\rangle + \vert 1\rangle-\vert 2\rangle)_{2},\\
		\vert \phi_{6}\rangle &= \frac{1}{\sqrt{33}}(\vert 1\rangle + \vert 2\rangle-\vert 3\rangle)_{1}(\vert 0\rangle -3 \vert 1\rangle-\vert 2\rangle)_{2}.\\
	\end{aligned}\label{eq8}
\end{equation}
\noindent\textbf{Proof.}
First, we demonstrate that the set of states in Eqs.~(\ref{eq8}) is indistinguishable by LOCC regardless of which party performs the initial measurement.

Since the maximum local dimension of the first party’s subsystem is four, we assume that Alice performs a POVM, wherein the POVM elements $M_{A}^{\dag}M_{A}$ can be written as 
\[
M^\dagger_{A}M_{A}=\begin{pmatrix}
	m_{00} & m_{01} & m_{02} & m_{03}\\
	m_{10} & m_{11} & m_{12} & m_{13}\\
	m_{20} & m_{21} & m_{22} & m_{23}\\
	m_{30} & m_{31} & m_{32} & m_{33}\\
\end{pmatrix}.
\]

Since any two states that are orthogonal only on the first subsystem are required to maintain orthogonality after measurement, we have the following results. For $|\phi_{1}\rangle$ and $|\phi_{4}\rangle$, we have $\langle\phi_{1}|M^\dagger_{A}M_{A}\otimes I_{B}|\phi_{4}\rangle=0$ and $\langle\phi_{4}|M^\dagger_{A}M_{A}\otimes I_{B}|\phi_{1}\rangle=0$. Thus we have $m_{01}=m_{10}=0$. Similarly, we have $m_{02}=m_{20}=0$ by states $|\phi_{1}\rangle$ and $|\phi_{5}\rangle$; $m_{12}=m_{21}=0$ by states $|\phi_{4}\rangle$ and $|\phi_{5}\rangle$; $m_{03}=m_{30}=0$ by states $|\phi_{1}\rangle$ and $|\phi_{6}\rangle$; $m_{13}=m_{31}=0$ by states $|\phi_{3}\rangle$ and $|\phi_{4}\rangle$, and $m_{23}=m_{32}=0$ by states $|\phi_{2}\rangle$ and $|\phi_{5}\rangle$. For states $|\phi_{3}\rangle$ and $|\phi_{6}\rangle$, we have $\langle\phi_{3}|M^\dagger_{A}M_{A}\otimes I_{B}|\phi_{6}\rangle=0$ and $\langle\phi_{6}|M^\dagger_{A}M_{A}\otimes I_{B}|\phi_{3}\rangle=0$. Thus we have $m_{22}=m_{33}$. Similarly, we have $m_{00}=m_{33}$ by states $|\phi_{2}\rangle$ and $|\phi_{3}\rangle$; and 
$m_{11}=m_{33}$ by states  $|\phi_{2}\rangle$ and $|\phi_{6}\rangle$. Thus, we have 
\[
M^\dagger_{A}M_{A}=\begin{pmatrix}
	m_{33} & 0 & 0 & 0\\
	0 & m_{33} & 0 & 0\\
	0 & 0 & m_{33} & 0\\
	0 & 0 & 0 & m_{33}\\
\end{pmatrix} \propto I.
\]
Hence,  Alice cannot perform a non-trivial measurement when she goes first.

  On the other hand, since the maximum local dimension of the second subsystem is three, we assume that Bob performs a POVM, wherein the POVM elements $M^\dagger_{B}M_{B}$ can be written as 
\[
M^\dagger_{B}M_{B}=\begin{pmatrix}
	M_{00} & M_{01} & M_{02}\\
	M_{10} & M_{11} & M_{12}\\
	M_{20} & M_{21} & M_{22}\\
\end{pmatrix}.
\]

Similarly, since any two states that are orthogonal only on the send subsystem are needed to maintain orthogonality after measurement to enable discrimination process to be carried out continuously, we have the following results. For $|\phi_{1}\rangle$ and $|\phi_{2}\rangle$, we have $\langle\phi_{1}|I_{A}\otimes M^\dagger_{B}M_{B}|\phi_{2}\rangle=0$ and $\langle\phi_{2}|I_{A}\otimes M^\dagger_{B}M_{B}|\phi_{1}\rangle=0$. Thus $M_{01}=M_{10}=0$. Similarly, we have $M_{02}=M_{20}=0$ by $|\phi_{1}\rangle$ and $|\phi_{3}\rangle$, and $M_{12}=M_{21}=0$ by $|\phi_{2}\rangle$ and $|\phi_{4}\rangle$. For states $|\phi_{3}\rangle$ and $|\phi_{5}\rangle$, we have $\langle\phi_{3}|I_{A}\otimes M^\dagger_{B}M_{B}|\phi_{5}\rangle=0$ and $\langle\phi_{5}|I_{A}\otimes M^\dagger_{B}M_{B}|\phi_{3}\rangle=0$. Thus $M_{11}=M_{22}=0$. Similarly, for states $|\phi_{4}\rangle$ and $|\phi_{5}\rangle$, we have $\langle\phi_{4}|I_{A}\otimes M^\dagger_{B}M_{B}|\phi_{5}\rangle=0$ and $\langle\phi_{5}|I_{A}\otimes M^\dagger_{B}M_{B}|\phi_{4}\rangle=0$. Thus $M_{00}=M_{22}=0$. 
So, we have
\[
M^\dagger_{B}M_{B}=\begin{pmatrix}
	M_{00} & 0 & 0\\
	0 & M_{00} & 0\\
	0 & 0 & M_{00}\\
\end{pmatrix} \propto I
\]
Hence,  Bob cannot perform a non-trivial measurement when he goes first.

Therefore, the set of six bipartite OPSs in Eq. (\ref{eq8}) cannot be perfectly distinguished by LOCC. This completes the proof.\hfill\rule{6pt}{6pt}

  From Eq. (\ref{eq8}), we know that there exists a set of six bipartite OPSs corresponding to the orthogonality graph of case (21) in Fig. \ref{fig7}, which cannot be locally distinguished. In fact, we only need to alter the first state in Eq. (\ref{eq8}) to obtain Eq. (\ref{eq9}), and the resulting new set of states can be locally distinguished with a certain probability.
  
  \begin{equation}
	\begin{aligned}
		\vert \phi^{\prime}_{1}\rangle &= \frac{1}{\sqrt{2}} \vert 0\rangle_{1}(\vert 0\rangle+|3\rangle)_{2},\\
		\vert \phi_{2}\rangle &= \frac{1}{\sqrt{3}} (\vert 0\rangle+\vert 1\rangle+\vert 3\rangle)_{1}\vert 1\rangle_{2},\\
		\vert \phi_{3}\rangle &=\frac{1}{\sqrt{6}} (\vert 0\rangle - \vert 2\rangle-\vert 3\rangle)_{1}(\vert 1\rangle + \vert 2\rangle)_{2},\\
		\vert \phi_{4}\rangle &=\frac{1}{\sqrt{2}}\vert 1\rangle_{1} (\vert 0\rangle + \vert 2\rangle)_{2},\\
		\vert \phi_{5}\rangle &= \frac{1}{\sqrt{6}}\vert 2\rangle_{1}(2\vert 0\rangle + \vert 1\rangle-\vert 2\rangle)_{2},\\
		\vert \phi_{6}\rangle &= \frac{1}{\sqrt{33}}(\vert 1\rangle + \vert 2\rangle-\vert 3\rangle)_{1}(\vert 0\rangle -3 \vert 1\rangle-\vert 2\rangle)_{2}.\\
	\end{aligned}\label{eq9}
\end{equation}

For the second party, Bob performs a POVM with operators $|3\rangle\langle3|$ and $I-|3\rangle\langle3|$. If the measured state is $|\phi^{\prime}_{1}\rangle$ (the probability is 1/6, assuming those six states are equally likely), it can be identified with the probability $ \frac{1}{2} (\langle 0|+\langle3|) |3\rangle\langle3| (\vert 0\rangle+|3\rangle)=\frac{1}{2}$.

\begin{center} 5. Case (24) in Fig. \ref{fig8} \end{center}

In Eq. (\ref{eq10}), we present a set of six bipartite OPSs that correspond to case (24) in Fig. \ref{fig8}.

\begin{equation}
	\begin{aligned}
		\vert \psi_{1}\rangle &= \frac{1}{2}(\vert 0\rangle+\vert 1\rangle+\vert 2\rangle+\vert 3\rangle)_{1}\vert 1\rangle_{2},\\
        \vert \psi_{2}\rangle &= \frac{1}{\sqrt{3}}\vert 2\rangle_{1}(\vert 0\rangle+ \vert 2\rangle)_{2},\\
        \vert \psi_{3}\rangle &= \frac{1}{\sqrt{3}}\vert 1\rangle_{1}(\vert 0\rangle + \vert 2\rangle+|3\rangle)_{2},\\
		\vert \psi_{4}\rangle &= \vert 0\rangle_{1}\vert 0\rangle_{2},\\
        \vert \psi_{5}\rangle &= \frac{1}{3\sqrt{2}}(\vert 1\rangle+\vert 2\rangle-2|3\rangle)_{1}(\vert 0\rangle +\vert 1\rangle-\vert 2\rangle)_{2},\\
		\vert \psi_{6}\rangle &= \frac{1}{2}(\vert 0\rangle - \vert 3\rangle)_{1}(\vert 1\rangle + \vert 2\rangle)_{2}.\\
	\end{aligned}\label{eq10}
\end{equation}

Similar to the derivation of Eq. (\ref{eq9}) above, for the set  in Eq. (\ref{eq10}), Bob performs a POVM with operators $|3\rangle\langle3|$ and $I-|3\rangle\langle3|$. If the measured state is $|\psi_{3}\rangle$ (the probability is 1/6, assuming those six states are equally likely), it can be identified with the probability $ \frac{1}{3} (\langle 0|+\langle2|+\langle3|) |3\rangle\langle3| (\vert 0\rangle+|2\rangle+|3\rangle)=\frac{1}{3}$.

\section{Conclusions}\label{sec6}

The local distinguishability of OPSs constitutes a research topic worthy of in-depth exploration. Up to now, the local distinguishability of four or five bipartite OPSs has been well addressed, whereas that of six bipartite OPSs has not yet been effectively solved. In this paper, We take into account all possible sets formed by six bipartite orthogonal product states (OPSs). After eliminating configurations that are identical under party swapping or vertex relabeling, there exist a total of 78 distinct orthogonality graphs for six bipartite OPSs. We prove that the sets of six bipartite OPSs corresponding to 73 of these orthogonality graphs are locally distinguishable. Furthermore, we discuss the local distinguishability properties for the remaining five cases separately.  For a set of six bipartite OPSs corresponding to three of these five cases, it cannot be perfectly distinguished by LOCC, or can be locally distinguished with a certain probability. For the remaining two cases, the first admits both a set of six bipartite OPSs that are locally indistinguishable and sets of states that can be locally discriminated with a certain probability, whereas we only identify probabilistically locally distinguishable state sets in the second case. Our work provides theoretical support for judging the local distinguishability of six bipartite OPSs.
\begin{acknowledgments}
	\vspace{-10pt}
    This work is supported by Natural Science Foundation of Shandong Province of China (Grant No. ZR2023MF080) and Beijing Natural Science Foundation (Grant No. 4252014).
\end{acknowledgments}

\nocite{*}
\bibliography{apssamp}

\begin{thebibliography}{0}%
\makeatletter
\providecommand \@ifxundefined [1]{%
 \@ifx{#1\undefined}
}%
\providecommand \@ifnum [1]{%
 \ifnum #1\expandafter \@firstoftwo
 \else \expandafter \@secondoftwo
 \fi
}%
\providecommand \@ifx [1]{%
 \ifx #1\expandafter \@firstoftwo
 \else \expandafter \@secondoftwo
 \fi
}%
\providecommand \natexlab [1]{#1}%
\providecommand \enquote  [1]{``#1''}%
\providecommand \bibnamefont  [1]{#1}%
\providecommand \bibfnamefont [1]{#1}%
\providecommand \citenamefont [1]{#1}%
\providecommand \href@noop [0]{\@secondoftwo}%
\providecommand \href [0]{\begingroup \@sanitize@url \@href}%
\providecommand \@href[1]{\@@startlink{#1}\@@href}%
\providecommand \@@href[1]{\endgroup#1\@@endlink}%
\providecommand \@sanitize@url [0]{\catcode `\\12\catcode `\$12\catcode
  `\&12\catcode `\#12\catcode `\^12\catcode `\_12\catcode `\%12\relax}%
\providecommand \@@startlink[1]{}%
\providecommand \@@endlink[0]{}%
\providecommand \url  [0]{\begingroup\@sanitize@url \@url }%
\providecommand \@url [1]{\endgroup\@href {#1}{\urlprefix }}%
\providecommand \urlprefix  [0]{URL }%
\providecommand \Eprint [0]{\href }%
\providecommand \doibase [0]{https://doi.org/}%
\providecommand \selectlanguage [0]{\@gobble}%
\providecommand \bibinfo  [0]{\@secondoftwo}%
\providecommand \bibfield  [0]{\@secondoftwo}%
\providecommand \translation [1]{[#1]}%
\providecommand \BibitemOpen [0]{}%
\providecommand \bibitemStop [0]{}%
\providecommand \bibitemNoStop [0]{.\EOS\space}%
\providecommand \EOS [0]{\spacefactor3000\relax}%
\providecommand \BibitemShut  [1]{\csname bibitem#1\endcsname}%
\let\auto@bib@innerbib\@empty
\end{thebibliography}%


\begin{thebibliography}{}\label{sec:TeXbooks}
	\bibitem{1}
	C. H. Bennett, D. P. DiVincenzo, C. A. Fuchs, T. Mor, E. Rains, P. W. Shor, J. A. Smolin, and W. K. Wootters. Quantum nonlocality without entanglement. Phys. Rev. A 59, 1070 (1999).
    \bibitem{J2000}          
    J. Walgate, A. J. Short, L. Hardy, and V. Vedral, Local distinguishability of multipartite orthogonal quantum states, Phys. Rev. Lett. 85, 4972 (2000).
    \bibitem{J2002}
    J. Walgate and L. Hardy, Nonlocality, asymmetry, and distinguishing bipartite states, Phys. Rev. Lett. 89, 147901 (2002).
    \bibitem{DiVincenzo2003}
    D. P. DiVincenzo, T. Mor, P. W. Shor, J. A. Smolin, and B. M. Terhal, Unextendible product bases, uncompletable product bases and bound entanglement, Commun. Math. Phys. 238, 379 (2003).
    \bibitem{Feng2009}
    Y. Feng and Y. Shi, Characterizing locally indistinguishable orthogonal product states, IEEE Trans. Inf. Theory 55, 2799–2806 (2009). 
    \bibitem{Duan2009}
    R. Y. Duan, Y. Feng, Y. Xin, and M. S. Ying, Distinguishability of quantum states by separable operations, IEEE Trans. Inf. Theory 55, 1320–1330 (2009).
    \bibitem{Duan2014}
    E. Chitambar and R. Duan, When do local operations and classical communication suffice for two-qubit state discrimination? IEEE Trans. Inf. Theory 60, 1549–1561 (2014).
       
	\bibitem{2}    
	Z. C. Zhang, F. Gao, G. J. Tian, T. Q. Cao and Q. Y. Wen. Nonlocality of orthogonal product basis quantum states. Phys. Rev. A 90, 022313 (2014).
   \bibitem{5}
	Z. C. Zhang, F. Gao, S. J. Qin, Y. H. Yang, and Q. Y. Wen, Nonlocality of orthogonal product states, Phys. Rev. A 92, 012332 (2015).
	\bibitem{6}
	Y. L. Wang, M. S. Li, Z. J. Zheng, and S. M. Fei, Nonlocality of orthogonal product-basis quantum states, Phys. Rev. A 92, 032313 (2015).
	\bibitem{7}
	Z. C. Zhang, F. Gao ,Y. Cao, S. J. Qin and Q. Y. Wen, Local indistinguishability of orthogonal product states. Phys. Rev. A 93, 012314 (2016).
	\bibitem{Cao2023}
    H. Q. Cao, M. S. Li, and H. J. Zuo, Locally stable sets with minimum cardinality, Phys. Rev. A 108, 012418 (2023).
    \bibitem{8}
	G. B. Xu, Q. Y. Wen, S. J. Qin, Y. H. Yang, and F. Gao, Quantum nonlocality of multipartite orthogonal product states, Phys. Rev. A 93, 032341 (2016).
	\bibitem{9}   
	Y. L. Wang, M. S. Li, Z. J. Zheng and S. M. Fei, The local indistinguishability of multipartite product states. Quantum Info. Process. 16, 5 (2017).
	\bibitem{10}
	G. B Xu, Q. Y. Wen, F. Gao, S. J. Qin and H. J. Zuo, Local indistinguishability of multipartite orthogonal product bases. Quantum Info. Process. 16, 276 (2017).
	\bibitem{11}
	Z. C. Zhang, K. J. Zhang, F. Gao, Q. Y. Wen, and C. H. Oh, Construction of nonlocal multipartite quantum states, Phys. Rev. A 95, 052344 (2017).
	
	\bibitem{13}
	S. Halder, Several nonlocal sets of multipartite pure orthogonal product states, Phys. Rev. A 98, 022303 (2018). 

    \bibitem{Halder2019}
     S. Halder, M. Banik, S. Agrawal, and S. Bandyopadhyay, Strong quantum nonlocality without entanglement, Phys. Rev. Lett. 122, 040403 (2019).
     
    \bibitem{Yuan2020}  
     P. Yuan, G. Tian, and X. Sun, Strong quantum nonlocality without entanglement in multipartite quantum systems, Phys. Rev. A 102, 042228 (2020).
    \bibitem{Shi2020}
     F. Shi, M. Hu, L. Chen, and X. Zhang, Strong quantum nonlocality with entanglement, Phys. Rev. A 102, 042202 (2020).
	\bibitem{WangYL2021}
     Y. L. Wang, M. S. Li, and M. H. Yung, Graph-connectivity based strong quantum nonlocality with genuine entanglement, Phys. Rev. A 104, 012424 (2021).
    \bibitem{FSHI2022}
     F. Shi, M. S. Li, L. Chen, and X. Zhang, Strong quantum nonlocality for unextendible product bases in heterogeneous systems, J. Phys. A: Math. Theor. 55, 015305 (2022).
    \bibitem{Fshi2022}
     F. Shi, Z. Ye, L. Chen, and X. Zhang, Strong quantum nonlocality in N-partite systems, Phys. Rev. A 105, 022209 (2022).
     \bibitem{Gao2023}
     H. Zhou, T. Gao, and F. Yan, Strong quantum nonlocality without entanglement in an n-partite system with even n, Phys. Rev. A 107, 042214 (2023).
    
	\bibitem{GBX2026}     
	G. B. Xu, Z. Y. Hao, H. K. Wang, Y. G. Yang, and D. H. Jiang, Local distinguishability of four multipartite orthogonal product states, Phys. Rev. A 112, 042433 (2025).
	
    \bibitem{XHWYJ}
	G. B. Xu, Z. Y. Hao, H. K. Wang, Y. G. Yang, and D. H. Jiang, Local distinguishability of five orthogonal product states on bipartite and tripartite quantum systems. Phys. Rev. A 113, 
     052436 (2026).
    \bibitem{Bondy2008}
    J. A. Bondy and U. S. R. Murty. Graph Theory[M]. Springer Publishing Company, Incorporated, 2008.
	\bibitem{21}
	M. A. Nielsen and I. L. Chuang, Quantum Computation and Quantum Information (Cambridge University, Cambridge, 2010)
\end{thebibliography}
\end{document}